%% file: main.tex
\documentclass{article}
\usepackage{ijcai26}

\usepackage{times}
\usepackage{soul}
\usepackage{url}
\usepackage[hidelinks]{hyperref}
\usepackage[utf8]{inputenc}
\usepackage[small]{caption}
\usepackage{graphicx}
\usepackage{amsmath}
\usepackage{amsthm}
\usepackage{booktabs}
\usepackage{algorithm}
\usepackage{algorithmic}
\usepackage[switch]{lineno}
\usepackage{multirow}
\usepackage{float}

\newtheorem{example}{Example}
\newtheorem{theorem}{Theorem}
\newtheorem{proposition}{Proposition}

\newtheorem{mechanism}{Mechanism}

    \title{Mechanism Design for Locating a Bridge Between Regions with Prelocated Facilities}

\author{
Genjie Qin$^1$
\and
Chenhao Wang$^2$
\and
Jianan Lin$^3$
\and
Qizhi Fang$^1$
\and
Wenjing Liu$^{1,*}$\\
\affiliations
$^1$School of Mathematical Sciences, Ocean University of China\\
$^2$Beijing Normal University\\
$^3$Rensselaer Polytechnic Institute\\
\emails
qgj@ouc.edu.cn,
chenhwang@bnu.edu.cn,
hcm6755@gmail.com,
qfang@ouc.edu.cn,
liuwj@ouc.edu.cn
}

\begin{document}
\maketitle
\begin{abstract}
In many urban planning projects, social planners require the construction of a bridge to connect two regions separated by obstacles such as rivers or highways. This paper studies the mechanism design problem for locating a bridge between two separate regions, each of which has been equipped with a facility. There are a set of agents located in each region and each agent has her location as private information. Once the bridge is built, the agents will go to the nearest facility to receive service and each agent's cost is the distance from her location to the nearest prelocated facility via the bridge. We investigate social cost and maximum cost under strategyproof (SP) mechanisms, with stronger notions of group-strategyproof (GSP) and strong group-strategyproof (SGSP). 

For the maximum cost objective, we characterize the optimal solution and show that it satisfies GSP. Under the SGSP, we propose a deterministic 3-approximation mechanism and a randomized 2-approximation mechanism, while proving a lower bound of 2 for any deterministic SGSP mechanism. For the social cost objective, we present a deterministic 3-approximation mechanism and a randomized 2-approximation mechanism that satisfy GSP. We establish lower bounds of 2 and 1.1 for deterministic and randomized SP mechanisms, respectively. Under the SGSP, the lower bound for deterministic mechanisms increases to \(1 + \min\{m, n\}\), and we provide a \((1 + 2 \min\{m, n\})\)-approximation mechanism. For randomized mechanisms, the lower bound remains 1.1, while an upper bound of \((1 + 2mn\cdot (m+n)^{-1})\) is achieved.
\end{abstract}

\section{Introduction}
The facility location problem is an important problem in the field of optimization. Its main goal is to select the best locations to place facilities under given constraints to optimize certain objectives. In the facility location problem, each agent possesses private information that cannot be directly verified by the social planner. This informational asymmetry creates strategic incentives for agents to misreport their private information to influence the facility placement decision. 

In 2013, \cite{LW1} employed this problem to introduce their highly influential framework of approximate mechanism design without money.  In its standard formulation, the facility location problem tasks a social planner with locating one or two facilities on a real line to minimize either the maximum individual cost or the total social cost, based on the locations reported by agents. Since agents may strategically misreport their positions, the goal is to design mechanisms that are approximately optimal while ensuring truthfulness. The problem has since become a central topic of investigation in theoretical computer science and artificial intelligence. 

Most existing research focuses on locating facilities in unconnected spaces. While such models provide a foundational theoretical framework, they are often inadequate for addressing a prevalent class of real-world scenarios: two regions, separated by natural or artificial barriers, each already equipped with public service facilities, where a planner aims to construct a bridge to connect them. This enables agents to flexibly access the nearest facility—in either region—based on their real-time locations, thereby reducing travel distance, lowering time costs, and promoting equitable utilization of public resources.

The essence of this ``bridge location problem" is fundamentally a reconfiguration of spatial relationships between two regions. By strategically placing a bridge, we can strengthen the cooperative service capacity of existing facilities in both areas. The primary societal benefits include: direct reduction of physical distances and improved mobility; dynamic, proximity-based facility choice facilitated by the bridge, leading to an adaptive service network; and mutual complementarity between facilities across regions, boosting overall resource efficiency and fairness. For example, in a city with railway stations on both sides of a river, a well-positioned bridge allows some residents to bypass the station in their own region and directly access a closer station on the opposite bank, thus reducing commute times and alleviating localized traffic congestion.

Although bridge engineering optimization and network design have been preliminarily explored in operations research, systematic research on its location problem from a mechanism design perspective remains scarce~\cite{LW-1,LW-2}. Particularly under realistic constraints—such as agents holding private location information and the potential for strategic misreporting or coalitional manipulation—designing a bridge location mechanism that ensures truthful reporting, resists collusion, and maximizes social welfare represents a meaningful gap in the literature.

In the field of mechanism design, \cite{LW18} formally introduces the bridge location problem into the mechanism design framework, emphasizing how bridges can generate social value by integrating existing facilities across regions. We propose a practical model that further captures this value by specifying that following bridge construction, agents may autonomously select the nearest facility in either region, with costs determined by the shortest path.

The core of our model lies in a geometric simplification of the connection. We assume that the two regions to be connected are parallel, and the bridge is perpendicular to both regions with zero traversal cost. This simplification retains generality. First, representing the regions as parallel lines captures the essential notion of disconnection, abstracting away non‑essential geometric details such as shape and relative separation. Second, the physical length of a bridge is typically negligible compared to the scale of the regions; its primary function is to enable connectivity rather than incur travel cost. Furthermore, any realistic connection can be equivalently mapped to this perpendicular configuration through a suitable coordinate transformation. By adopting the parallel‑region, perpendicular zero‑cost bridge framework, the model preserves the essential structure of the problem while simplifying its geometric representation, allowing the analysis to focus on how the bridge‑induced connectivity influences the efficiency of resource allocation and social welfare.

This research revolves around the following key questions: in a flexible setting where agents can autonomously choose the nearest facility to use, how to design a bridge location mechanism that satisfies (strong) group-strategyproofness; what approximate efficiency such mechanisms can achieve under the objectives of minimizing maximum cost or social cost; and whether there exist theoretical performance limits.

\subsection{Related Work}
Our work is grounded in a series of studies for the design of approximate mechanisms without money, which was initiated by~\cite{LW1}. Most studies focused on the location problem where prelocated facilities are not taken into account. However, many situations require incorporating prelocated facilities to improve agents' accessibility and better reflect practical conditions.

\subsubsection{Facility Location without Prelocated Facilities.}  
In the one-facility setting, \cite{LW1} gave the best possible deterministic and randomized mechanisms under the objectives of maximum cost and social cost.~\cite{LW2} extended the facility location problem to circular and general graph structures and proposed a randomized mechanism with an approximation ratio of \(2 - 2/n\). Subsequent studies by~\cite{LW3},~\cite{LW4} and~\cite{LW5}. ~\cite{LW6} explored improved mechanisms and lower bounds in tree structures, discrete spaces, and strictly convex spaces. For the homogeneous two-facility,~\cite{LW1} proposed an optimal $2$-approximation under the maximum cost objective and a \((n - 2)\)-approximation under the social cost objective.~\cite{LW7} gave the lower bound to $n-2$ for deterministic mechanisms under the social cost, while~\cite{LW8} designed a randomized mechanism with a $4$-approximation in general metric spaces. In the heterogeneous two-facility setting,~\cite{LW9} introduced the dual preference model.~\cite{LW10,LW11} introduced the optional preference model, where the cost of each agent is the sum of the distances to the facilities they are interested in. \cite{LW12} investigated a setting in which agents select the nearest or farthest facility of interest.~\cite{LW13} considered a setting in which the cost of an agent is \(1 - d\) if she is interested in a facility at distance \(d\) and \(0\) otherwise. ~\cite{LW14} studied the fractional preference where the preference of each agent for the facilities is represented by a number between 0 and 1.

\subsubsection{Facility Location with Prelocated Facilities.}  
In many practical applications, prelocated resources
(such as facilities) should not be wasted. \cite{LW15} considered modifying the structure of the regions by adding a costless shortcut edge based on a prelocated facility on the real line. \cite{LW16} further studied the problem of adding a non-zero cost shortcut or two costless shortcuts. \cite{LW17} studied a basic scenario in which an obstacle divides the interval $[0,1]$ into two regions. Connectivity was restored by building bridges, with traversal costs proportional to bridge length. Agents on the left incurred costs equal to the distance to endpoint 1 via the bridge, while those on the right incurred costs to endpoint 0. 

\cite{LW18} studied the problem of constructing a bridge that connects two parallel regions, each of which has been equipped with a facility.  Our work differs fundamentally from the study by \cite{LW18}. While their setting assumed agents are located in separate regions and must cross a bridge to reach a facility in the opposite region, we allow agents to autonomously choose to access the closer facility in either of the two regions via a bridge, with their costs determined by the shortest path distance to their nearest facility. For example, in a city divided by a river, stations are located on both banks, and residents could select the closer railway station located in either region to reduce travel costs. Allowing agents to autonomously select the nearest facility based on distance makes the model more aligned with practical scenarios.

\subsection{Our Results}
In this paper, we study the mechanism design problem for locating a bridge between two separate regions, each of which has been equipped with a facility. We derive upper and lower bounds on the achievable approximation ratio for strategyproof and (strong) group-strategyproof mechanisms under the objectives of minimizing the maximum cost and the social cost. The results are summarized in Table \ref{tab1}, where $[a,b]$ represent the upper bound and the lower bound, respectively.


\begin{table}[htbp]
\centering
\renewcommand{\arraystretch}{1.3} 
\begin{tabular}{ccc}
\hline
\textbf{Objective} & \textbf{Maximum Cost} & \textbf{Social Cost} \\
\hline
\multirow{2}{*}{\textbf{SP/GSP}} 
    & Det.: 1& Det.: [2,3] \\
    & Rand.: 1 & Rand.: [1.1,2] \\
\hline
\multirow{2}{*}{\textbf{SGSP}} 
    & Det.: [2,3]& Det.: $[1+k,1+2k]$ \\
    & Rand.: [1,2] & Rand.: $\left[1.1,1+\frac{2mn}{m+n}\right]$ \\
\hline
\end{tabular}
\caption{\parbox{\linewidth}{\raggedright Summary of results. Each interval $[\text{lower bound}, \text{upper bound}]$ shows the approximation ratio bounds. ($k=\min\{m,n\}$).}}
\label{tab1}
\end{table}

\section{Model}
Let $L_1, L_2$ be two parallel lines that are used to indicate two regions separated by a natural or infrastructural barrier. Each line is equipped with a facility, located at $x_F$ (on $L_1$) and $y_F$ (on $L_2$), respectively.  There are agents distributed along the two lines, with $N = \{1,...,n\}$ on line $L_1$ and $M = \{1,...,m\}$ on line $L_2$. For consistency, we parameterize both lines using the real number line $R$, assigning each point a coordinate in $R$. Each agent $i \in N$ has a private location $x_i \in R$, and each agent $j \in M$ has a private location $y_j \in R$. Let $(\mathbf{x},\mathbf{y})$ denote the location profile (or instance), where $\mathbf{x}=(x_1,...,x_n)$ and $\mathbf{y}=(y_1,...,y_m)$. 

$L_1$ and $L_2$ are connected by a perpendicular bridge with zero travel cost, positioned at a point $s \in R$. Note that when $x_{F} = y_{F}$, the trivial mechanism that outputs \(x_F\) as the location of the bridge is clearly strategyproof and optimal under both social objectives. In the remainder of this paper, we consider the normalized scenario where $x_F \neq y_F$ with $x_F=1$ and $y_F=0$. This normalization preserves generality since locations can be arbitrarily scaled on the real line.

In our model, the facilities $F_1$ and $F_2$ are functionally identical, providing perfectly substitutable services to all agents. Each agent chooses between accessing their local facility directly or crossing the bridge to reach the opposite facility, whichever minimizes their travel distance. 

A deterministic mechanism is a function $f : R^{n}\times R^{m}\rightarrow R$ that maps the location profiles \( \mathbf{x} \) and \( \mathbf{y} \) of all agents to a real number, representing the location of the bridge. If $f(\textbf{x},\textbf{y})=s$,  the costs of each agent in sets \( N \) and \( M \) are defined as
\begin{align*}
c_1(x_i,s) &:= \min\{|x_i - x_F|, |x_i - s| + |s - y_F|\}, \quad  i \in N \\
c_2(y_j,s) &:= \min\{|y_j - y_F|, |y_j - s| + |s - x_F|\}, \quad  j \in M
\end{align*}

A randomized mechanism is a function from \(R^n \times R^m \) to a probability distribution over real numbers.  For a randomized mechanism $f$, the costs of each agent in sets \( N \) and \( M \) are defined as
\[
c_1(x_i, f(\mathbf{x},\mathbf{y})) := E_{s \sim f}[c_1(x_i, s)], \quad i \in N,
\]
\[
c_2(y_j, f(\mathbf{x},\mathbf{y})) := E_{s \sim f}[c_2(y_j, s)], \quad j \in M.
\]

A mechanism is strategyproof if no agent can decrease their cost by misreporting their location within their designated region regardless of the locations reported by the others. Formally, 
$\forall (\mathbf{x},\mathbf{y})\in R^{n}\times R^m$, 
$\forall i\in N$, $\forall x_{i}'\in R$, 
$\forall j\in M$, $\forall y_{j}'\in R$, 
we have 
$c_1(x_{i},f(\mathbf{x},\mathbf{y}))\leq c_1\left(x_{i},f(x_{i}',(\mathbf{x},\mathbf{y})_{-i})\right)$ 
and 
$c_2(y_{j},f(\mathbf{x},\mathbf{y}))\leq c_2(y_{j},f(y_{j}',(\mathbf{x},\mathbf{y})_{-j}))$. 
Here, $(\mathbf{x},\mathbf{y})_{-i}$ represents the location profile of the set of agents $N\cup M\!\setminus\!\{i\}$, $i\in N$; the definition of $(\mathbf{x},\mathbf{y})_{-j}$ is similar.

A mechanism is group-strategyproof if for any location profile $(\mathbf{x},\mathbf{y})$ and any coalition $S \subseteq N\cup M$ with $(\mathbf{x},\mathbf{y})_S$, there is no joint misreport of location $(\mathbf{x},\mathbf{y})_S'$ of the agents in $S$ such that all agents in $S$ gain. A mechanism is strong group-strategyproof if no group of agents can misreport in a way that at least one member strictly benefits while no other member’s cost increases.

The maximum cost and social cost of a mechanism $f$ with respect to a location profile \( (\mathbf{x},\textbf{y})\) are defined as the maximum cost among all $n+m$ agents and the total cost of all $n+m$ agents, respectively. For a deterministic mechanism $f$,
\begin{align*}
    &MC((\mathbf{x}, \mathbf{y}), f) \\&= \max\left\{\max_{i \in N} c_{1}(x_{i}, f(\textbf{x},\textbf{y})), \max_{j \in M} c_{2}(y_{j}, f(\textbf{x},\textbf{y}))\right\},\\
    &SC((\mathbf{x}, \mathbf{y}),f)\\ &= \sum_{i \in N} c_{1}(x_{i}, f(\textbf{x},\textbf{y})) + \sum_{j \in M} c_{2}(y_{j}, f(\textbf{x},\textbf{y})).
\end{align*}
For a randomized mechanism $f$, $$MC( (\mathbf{x},\textbf{y}),f)=\mathbf{E}_{Y\sim f( \mathbf{x},\textbf{y})}[MC((\textbf{x},\textbf{y}),Y)],$$ $$SC( (\mathbf{x},\textbf{y}),f)=\mathbf{E}_{Y\sim f( \mathbf{x},\textbf{y})}[SC((\textbf{x},\textbf{y}),Y)].$$

For an instance $(\mathbf{x}, \mathbf{y})$, let $s^{*}_{\text{MC}}(\textbf{x},\textbf{y})$ and $s^{*}_{\text{SC}}(\textbf{x},\textbf{y})$ denote its optimal solutions for the objective of maximum cost and social cost, respectively. The subscript is omitted without confusion. A strategyproof mechanism \( f \) has an approximation ratio of $\gamma\geq 1$ under the maximum cost objective, if 
\begin{equation*}
\gamma=\sup_{(\mathbf{x},\mathbf{y})\in R^{n}\times R^{m}}\frac{MC((\mathbf{x,y}),f)}{MC((\mathbf{x,y}),s^{*}_{MC})}.
\end{equation*} 
The approximation ratio is defined similarly under the social cost objective. In this paper, we will focus on anonymous\footnote{A mechanism is anonymous if its outcome depends only on the agents' locations, not on their identities.} strategyproof mechanisms with good approximation ratios under the social objectives of minimizing the maximum cost or minimizing the social cost. Henceforth, for simplicity, the instance $(\mathbf{x},\mathbf{y})$ is assumed to satisfy $x_{1}\leq x_{2}\leq...\leq x_{n},y_{1}\leq y_{2}\leq...\leq y_{m}$ without further comment.\\
\textbf{Notations.} Let \( lt(\mathbf{x}) \) and \( rt(\mathbf{x}) \) $\left(lt(\mathbf{y})\text{ and } rt(\mathbf{y})\right)$ denote the locations of the leftmost and rightmost agents in \( \mathbf{x} \) $(\mathbf{y})$, respectively. To facilitate the analysis, we partition the agents based on their relative positions. We first divide agents on each line at $\frac{1}{2}$. Then, we further partition agents by comparing their positions to the facility’s location on the opposite line, separating those on each side of that facility. We partition the agents as follows. Let $N_1$ denote the set of agents in $N$ with $x_i < 0.5$, and $N_2$ those with $x_i \geq 0.5$. Similarly, let $M_1$ be the set of agents in $M$ with $y_i \leq 0.5$, and $M_2$ those with $y_i > 0.5$. Denote $n_1 = |N_1|$, $n_2 = |N_2|$, $m_1 = |M_1|$, and $m_2 = |M_2|$. These partitions are illustrated in Fig. \ref{fig2}.

The mechanisms in this paper are analyzed under three incentive constraints: SP, GSP, and SGSP, where SGSP implies GSP, and GSP implies SP. Therefore, SP lower bounds hold for all three classes, and any upper bound achieved under a stronger constraint also applies to weaker ones.
\begin{figure}
    \centering
    \includegraphics[width=1\linewidth]{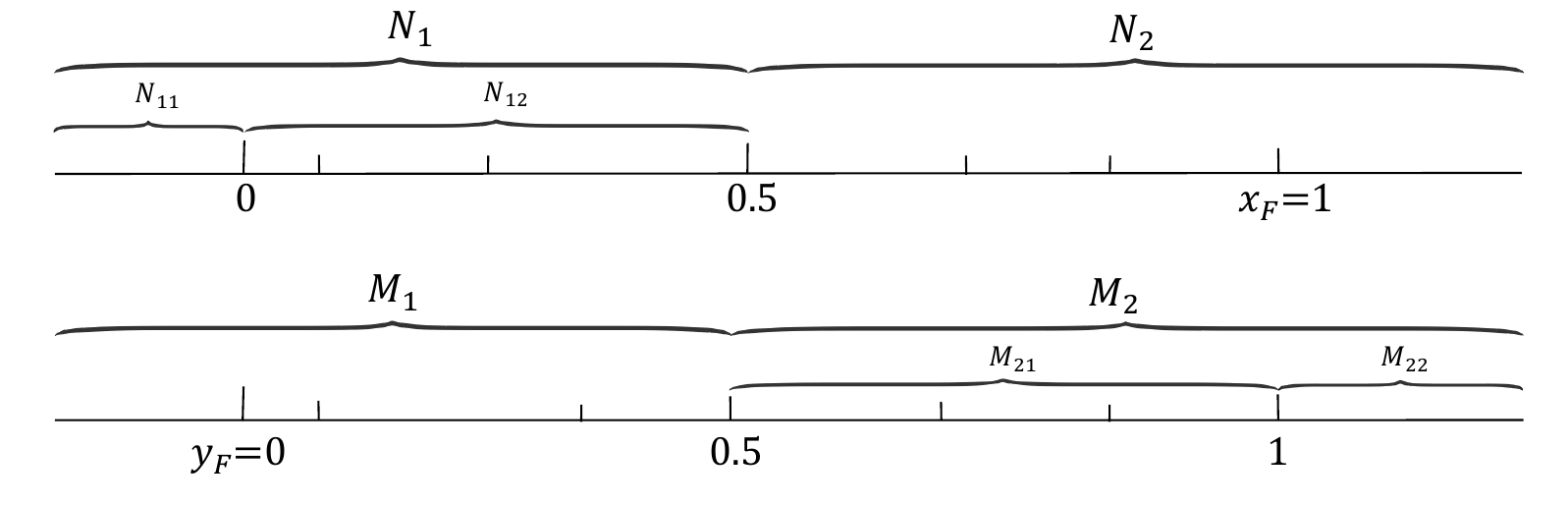}
    \caption{Partition of sets $N$ and $M$ based on $(\mathbf{x},\mathbf{y})$.}
    \label{fig2}
\end{figure}
\section{Maximum Cost}
This section considers the mechanism design for locating the bridge between two parallel lines with prelocated facilities, under the social objective of minimizing the maximum cost. We first present a simple optimal solution.
\begin{proposition}
    \label{thm1}
For any instance $(\mathbf{x},\mathbf{y)}$, an optimal solution for minimizing the maximum cost is given by  \begin{equation*}
    s^{*}( \mathbf{x},\mathbf{y})=\begin{cases}
        0, & \mbox{if }\textbf{ }1-lt(\mathbf{x})\geq rt(\mathbf{y}), \\
        1, & \mbox{if }\textbf{ }1-lt(\mathbf{x})< rt(\mathbf{y}).
      \end{cases}
  \end{equation*}
\end{proposition}
\vspace{0.5em}
Due to space limitations, some proofs are deferred to Appendix~\ref{app:supplementary}.
\vspace{0.5em}
\subsection{GSP Mechanism for Maximum Cost}
Note that the optimal solution is not strategyproof. For example, consider an instance $(x_{1},y_{1})$, where $x_{1}=0.25$ and $y_{1}=1$. Clearly, $s^{*}(x_{1},y_{1})=1$ by Proposition \ref{thm1}. Suppose agent 1 misreports her location $x_1=0.25$ to $x_1'=0$, resulting in instance $(x_1',y_1)$ with $s^{*}(x_1',y_1)=0$. Obviously, agent 1 can benefit from misrepresenting her location. To address this issue, we develop a distance-symmetric mechanism that maintains approximate optimal performance while ensuring that no agent can benefit from misrepresenting private information.

\begin{mechanism}\label{mec1}
  For any instance $(\mathbf{x,y})$, if $lt(\mathbf{x})\geq\frac{1}{2}$ or $rt(\mathbf{y})\leq\frac{1}{2}$, the bridge is located at $s$ as follows:
  \begin{equation*}
    s=\begin{cases}
        0, & \mbox{if }\textbf{ }rt(\mathbf{y})<\frac{1}{2}, \\
        1, & \mbox{if }\textbf{ }lt(\mathbf{x})\geq\frac{1}{2} .
      \end{cases}
  \end{equation*}
  If $lt(\mathbf{x})<\frac{1}{2}< rt(\mathbf{y})$, the bridge is located at $s$ as follows:
  \begin{equation*}
    s=\begin{cases}
        \frac{1}{2}-rt(\mathbf{y}), & \mbox{if }\text{ }1-lt(\mathbf{x})\geq rt(\mathbf{y}), \\
        \frac{3}{2}-lt(\mathbf{x}), & \mbox{if }\text{ }1-lt(\mathbf{x})<rt(\mathbf{y}).
      \end{cases}
  \end{equation*}
\end{mechanism}

\begin{theorem}
Mechanism \ref{mec1} is GSP and achieves optimality under the maximum cost objective
\end{theorem}
\begin{proof}
We first prove that the Mechanism \ref{mec1} is GSP. Suppose a coalition of agents misreports their locations to benefit all members. Note that agents in \( N_2 \) and \( M_1 \) never use the bridge and thus have no incentive to participate in misreporting. 

For any instance $(\mathbf{x},\mathbf{y})$, if $lt(\mathbf{x})\geq\frac{1}{2}$ or $rt(\mathbf{y})\leq\frac{1}{2}$, the mechanism always outputs the optimal solution, so no group can misreport to strictly benefit everyone. Now consider the case where $lt(\mathbf{x})<\frac{1}{2}<rt(\mathbf{y})$ and $1-lt(\mathbf{x})\geq rt(\mathbf{y})$, the mechanism outputs the location $s=\frac{1}{2}-rt(\mathbf{y})$. For a coalition of agents $S$, suppose the coalition $S$ misreports their locations, resulting in the reported instance $(\mathbf{x}', \mathbf{y}')$ and the corresponding mechanism output $s'$.

\textbf{Case 1.} $S$ satisfies $S\cap N_1\neq \emptyset$ and $S\cap M_2\neq \emptyset$.
If $lt(\mathbf{x}')\geq\frac{1}{2}$ or $rt(\mathbf{y}')\leq\frac{1}{2}$, then $s'\in \{0,1\}$. It is easy to see that there always exists at least one agent whose cost does not strictly decrease.
If $1-lt(\mathbf{x}')\geq rt(\mathbf{y}')$, the mechanism outputs the location $s'=\frac{1}{2}-rt(\mathbf{y}')\leq0$. Clearly, the cost of agents in $S\cap M_2$ does not decrease.
If $1-lt(\mathbf{x}')< rt(\mathbf{y}')$, the mechanism outputs the location $s'=\frac{3}{2}-lt(\mathbf{x}')>1$. Clearly, the cost of agents in $S\cap N_1$ does not decrease.

\textbf{Case 2.} $S$ satisfies $S\cap N_1\neq \emptyset$ and $S\cap M_2=\emptyset$.
If $lt(\mathbf{x}')\geq\frac{1}{2}$, then $s'=1$. Clearly, the agents' costs do not strictly decrease.
If $1-lt(\mathbf{x}')\geq rt(\mathbf{y})$, then $s=s'$.
If $1-lt(\mathbf{x}')<rt(\mathbf{y})$, then $s'=\frac{3}{2}-lt(\mathbf{x}')\geq1$, and thus the agents' costs do not decrease.

\textbf{Case 3.} $S$ satisfies $S\cap N_1= \emptyset$ and $S\cap M_2\neq\emptyset$.
If $rt(\mathbf{y}')\leq\frac{1}{2}$, then $s'=0$. Clearly, the costs of all agents in $S$ do not strictly decrease.
If $1-lt(\mathbf{x})\geq rt(\mathbf{y}')$, then $s'=\frac{1}{2}-rt(\mathbf{y}')<\frac{1}{2}$, and the costs of all agents in $S$ do not strictly decrease.
If $1-lt(\mathbf{x})<rt(\mathbf{y}')$, then $s'=\frac{3}{2}-lt(\mathbf{x})$. For any agent $j\in S$ with \(y_{j} \in [\frac{1}{2}, rt(\mathbf{y})]\), we have
\begin{align*}
    c_{2}(y_{j},s) &= y_{j}, \\
    c_{2}(y_{j},s')&=\min\left\{y_{j}, \left|\frac{3}{2} - lt(\mathbf{x}) - y_{j}\right| + \frac{3}{2} - lt(\mathbf{x}) - 1\right\} \\
     &= \min\left\{y_{j}, 2 - 2lt(\mathbf{x}) - y_{j}\right\}.
\end{align*}
Since \(y_{j} \in [\frac{1}{2}, rt(\mathbf{y})]\), it follows that \(c_{2}( y_{j},s') = y_{j} = c_{2}( y_{j},s)\). Therefore, there always exists at least one agent who does not obtain a strict benefit.

The proof for the case $lt(\mathbf{x}) < \frac{1}{2} < lt(\mathbf{y})$ with $1 - lt(\mathbf{x}) \leq rt(\mathbf{y})$ is analogous. This establishes group-strategyproofness.

We now prove that Mechanism \ref{mec1} achieves optimality under the maximum cost objective. Let $f$ denote Mechanism 1. For any instance $(\mathbf{x},\mathbf{y})$, if \(lt(\mathbf{x})\geq \frac{1}{2}\) and $rt(\mathbf{y})>\frac{1}{2}$, it is clear that for any output, the cost remains identical across all agents except those with $y_{j}\geq\frac{1}{2},j\in M$. Then for each agent $j\in M$ with $y_j\geq\frac{1}{2}$,
\begin{align*}
    c_{2}(y_{j},s^{*}(\mathbf{x},\mathbf{y}))&=|y_j-s^{*}(\mathbf{x},\mathbf{y})|+|s^{*}(\mathbf{x},\mathbf{y})-1|\\
    &\geq|y_{j}-1|=c_2(y_j,1).
\end{align*}
    If \(lt(\mathbf{x})\geq \frac{1}{2}\) and $rt(\mathbf{y})\leq\frac{1}{2}$, any output is optimal. If \(rt(\mathbf{y})\leq \frac{1}{2}\), the proof is symmetry. Now consider \(lt(\mathbf{x}) < \frac{1}{2} < rt(\mathbf{y})\).

\textbf{Case 1.} \(1 - lt(\mathbf{x}) \geq rt(\mathbf{y})\).\\
Let
\begin{align*}
    a &= \max_{i \in N_{2}, j \in M_{1}} \left\{c_{1}(x_{i}, s), c_{2}(y_{j}, s)\right\} \\&= \max_{i \in N_{2}, j \in M_{1}} \left\{|x_{i} - 1|, |y_{j}|\right\}.
\end{align*}
Then, \(MC((\mathbf{x}, \mathbf{y}), s^{*}) \geq \max\{a,lt(\mathbf{x}), rt(\mathbf{y})\}\), and
\begin{align*}
    &MC((\mathbf{x}, \mathbf{y}), f) \\
    &= \max\left\{a,\mathop{c_{1}}\limits_{i \in N_{1}:x_i\leq f}(x_{i},f), \mathop{c_{1}}\limits_{i \in N_{1}:x_i>f}(x_{i},f), rt(\mathbf{y})\right\} \\
    &\leq \max\left\{a,lt(\mathbf{x}), \frac{1 - (\frac{1}{2} - rt(\mathbf{y})) + rt(\mathbf{y}) - \frac{1}{2}}{2}, rt(\mathbf{y})\right\} \\
    &= \max\left\{a,lt(\mathbf{x}), rt(\mathbf{y})\right\}.
\end{align*}
Therefore, the mechanism outputs the optimal solution.

\textbf{Case 2.} \(1 - lt(\mathbf{x}) < rt(\mathbf{y})\). The proof is similar.
\end{proof}

\subsection{SGSP Mechanism for Maximum Cost}
\begin{mechanism}\label{mec4}
For any instance $(\mathbf{x},\mathbf{y})$, if \( n_1 > 0 \), return \( s = 0 \); otherwise return \( s = 1 \).
\end{mechanism}
\begin{theorem}
Mechanism \ref{mec4} is a deterministic SGSP 3-approximation mechanism for the maximum cost.
\end{theorem}
\begin{theorem}
Any deterministic SGSP mechanism has an approximation ratio of at least $2$ for the maximum cost.
\end{theorem}
\begin{mechanism}\label{mec5}
For any instance $(\mathbf{x},\mathbf{y})$:
\begin{itemize}
    \item If \( n_1 = 0 \), return \( s = 1 \);
    \item Else if \( m_2 = 0 \), return \( s = 1 \);
    \item Else return \( s = 0 \) with probability \( \frac{1}{2} \) and \( s = 1 \) with probability \( \frac{1}{2} \).
\end{itemize}
\end{mechanism}
\begin{theorem}
Mechanism \ref{mec5} is a randomized SGSP 2-approximation mechanism for the maximum cost.
\end{theorem}

\section{Social Cost}
For the social objective of minimizing the social cost, we first provide the form of the optimal solution and the corresponding optimal value for any instance $(\mathbf{x}, \mathbf{y})$. Then, we derive upper and lower bounds on the achievable approximation ratio for (strong) group-strategyproof mechanisms.
\begin{proposition}
\label{t1}
For the social cost, at least one of the endpoints $s\in \{0,1\}$ is an optimal bridge location.
\end{proposition}
\subsection{GSP Mechanism for Social Cost}
In our mechanism design, the social cost objective is intrinsically tied to the population distribution. We therefore construct a deterministic group-strategyproof mechanism based on the cardinalities of $N_1$ and $M_2$.
\begin{mechanism}\label{mec2}
For any instance $(\mathbf{x},\mathbf{y})$, if $n_1\geq m_{2}$, output $s=0$; otherwise output $s=1$.
\end{mechanism}
\begin{theorem}\label{thm4}
Mechanism \ref{mec2} is a deterministic GSP $3$-approximation mechanism for the social cost.
\end{theorem}
\begin{proof}
Let $f$ denote Mechanism \ref{mec2}. For any instance $(\mathbf{x},\mathbf{y})$ and any coalition containing agent $i \in N_2$ or $j \in M_1$, it is evident that each agent $i\in N_{2}$ will choose the facility $x_{F}$, while each agent $j\in M_{1}$ will choose the facility $y_{F}$. Now consider a coalition $S\subseteq N_{1}\cup M_{2}$ with location profile $(\mathbf{x},\mathbf{y})_S$ and misreported profile $(\mathbf{x},\mathbf{y})'_S$. Without loss of generality, assume $f(\mathbf{x},\mathbf{y}) = 0$, implying $N_1 \neq \emptyset$.

\textbf{Case 1.} For each agent $i \in N_1 \cap S$, if $f((\mathbf{x},\mathbf{y})'_S,(\mathbf{x},\mathbf{y})_{-S}) \allowbreak
= 0$, then the cost of agent $i$ remains unchanged. If $f((\mathbf{x},\mathbf{y})'_S,(\mathbf{x},\mathbf{y})_{-S}) = 1$, then
$$c_{1}(x_{i},f(\mathbf{x},\mathbf{y})')=1-x_{i}\geq |x_{i}|=c_1(x_{i},f(\mathbf{x},\mathbf{y})).$$

\textbf{Case 2.} For each agent $j \in N_2 \cap S$, if $f((\mathbf{x},\mathbf{y})'_S,\mathbf{x},\mathbf{y})_{-S}) \allowbreak 
= 0$, the cost of agent $j$ is unchanged. If $f((\mathbf{x},\mathbf{y})'_S(\mathbf{x},\mathbf{y})_{-S})\allowbreak
= 1$, there exists $i \in N_1 \cap S$ who cannot benefit from misreporting her locations. Hence, Mechanism \ref{mec2} is group-strategyproof. We now analyze the approximation ratio of Mechanism \ref{mec2}.

For any instance \((\mathbf{x}, \mathbf{y})\), let ALG and OPT denote the social costs of the mechanism and the optimal solution, respectively. If \(n_1 \cdot m_{2} = 0\), the mechanism clearly outputs the optimal solution. Without loss of generality, assume that \(n_1 \geq m_{2} > 0\) and \(f(\mathbf{x}, \mathbf{y}) = 0\). From Theorem \ref{t1}, we have
\[
s^{*}(\mathbf{x}, \mathbf{y}) \in \{0, 1\}.
\]
If \(s^{*}(\mathbf{x},\mathbf{y}) = 1\),  we have 
\begin{align*}
&\displaystyle\sum_{i \in N_{1}} |1 - x_{i}| + \displaystyle\sum_{j \in M_{2}} |1 - y_{j}|\\
=&  \displaystyle\sum_{i \in N_{11}} |x_{i}| + n_{11} + n_{12} - \displaystyle\sum_{i \in N_{12}} x_{i}\\
&+ m_{21} - \displaystyle\sum_{j \in M_{21}} y_{j} + \displaystyle\sum_{j \in M_{22}} y_{j} - m_{22}
\\
\geq&\displaystyle\sum_{i \in N_{11}} |x_{i}| + n_{11} + \frac{ n_{12}}{2} + \displaystyle\sum_{j \in M_{22}} y_{j} - m_{22}.
\end{align*}
Then, 
\begin{align*}
    \frac{ALG}{OPT} &\leq \frac{\displaystyle\sum_{i \in N_{11}} |x_{i}| + \displaystyle\sum_{i \in N_{12}} x_{i} + \displaystyle\sum_{j \in M_{21}} y_{j} + \displaystyle\sum_{j \in N_{22}} y_{j}}{\displaystyle\sum_{i \in N_{1}} |1 - x_{i}| + \displaystyle\sum_{j \in M_{2}} |1 - y_{j}|} \\
  &\leq \frac{\displaystyle\sum_{i \in N_{11}} |x_{i}| + \frac{ n_{12}}{2} + m_{21} + m_{22} + \displaystyle\sum_{j \in M_{22}} y_{j} - m_{22}}{\displaystyle\sum_{i \in N_{11}} |x_{i}| + n_{11} + \frac{ n_{12}}{2} + \displaystyle\sum_{j \in M_{22}} y_{j} - m_{22}} \\
    &\leq \frac{\frac{ n_{12}}{2} + m_{21} + m_{22}}{ n_{11} + \frac{ n_{12}}{2}} \leq \frac{\frac{ n_{12}}{2} + m_{2}}{n_1 - \frac{ n_{12}}{2}} \leq \frac{\frac{n_1}{2} + m_{2}}{\frac{n_1}{2}} \leq 3.
\end{align*}
If \( m_{2} > n_1 > 0\), the proof is similar. The upper bound is shown to be tight by Example \ref{eg3}.
\begin{example}\label{eg3}
Consider the instance \((\textbf{x},\textbf{y})=(x_{1}, y_{1})\), where \(x_{1} = \frac{1}{2}-\epsilon\) and \(y_{1} = 1\) where $\epsilon\rightarrow0^+$. The optimal social cost is \(OPT= \frac{1}{2}+\epsilon,\) the social cost of Mechanism \ref{mec2} is \( ALG=\frac{3}{2}-\epsilon\). So the ratio approaches 3 as $\epsilon\rightarrow0+$.
\end{example}
This completes the proof of the approximation ratio.
\end{proof}
\begin{proposition}
If a deterministic SP mechanism only considers \( n_1, n_2, m_1, m_2 \) instead of the complete agents' profile, then no such mechanism is better than 3-approximation.
\end{proposition}
\begin{theorem}
    Any deterministic SP mechanism has an approximation ratio of at least $2$ for the social cost.
\end{theorem}
\begin{proof}
Suppose there is an SP mechanism with approximation ratio better than 2. Consider an instance $(\mathbf{x,\mathbf{y}})=(x_1,y_1)$ with $x_1 = 0$ and $y_1 = 1$. An optimal solution is $s^* = 0$ and the optimal social cost is $1$. Due to symmetry, we have three cases, assuming w.l.o.g. that the solution returned by the mechanism is $s \leq \frac{1}{2}$.

\textbf{Case 1.} $s = \frac{1}{2}$ or $s \leq -\frac{1}{2}$. The social cost induced by the solution $s$ is $2$ and the optimal social cost is $1$, indicating an approximation ratio of at least $2$.

\textbf{Case 2.} $0 \leq s < \frac{1}{2}$. We move the agent location $x_1 = 0$ to $x'_1 = \frac{1}{2} - \epsilon$ for some sufficiently small positive number $\epsilon < \frac{1}{2} - s$. Let $s'$ be the solution returned by the mechanism for the instance with $x'_1 = \frac{1}{2} - \epsilon$ and $y_1 = 1$. Note that $s'$ cannot be $s' \geq \frac{1}{2}$, as otherwise the agent at $x'_1$ can misreport from $\frac{1}{2} - \epsilon$ to $0$, which will decrease the cost. Hence, $s' < \frac{1}{2}$, and the social cost of the mechanism is at least $1 + \frac{1}{2} - \epsilon$. However, the optimal social cost is $\frac{1}{2} + \epsilon$ with an optimal bridge location $1$, indicating an approximation ratio approaching $3$.

\textbf{Case 3.} $-\frac{1}{2} < s < 0$. We move the agent location $x_1 = 0$ to $x''_1 = \frac{1}{2} + s - \epsilon$ for some small positive number $\epsilon$. Consider the instance with $x''_1 = \frac{1}{2} + s - \epsilon$ and $y_1 = 1$. The optimal social cost is $\frac{1}{2} - s + \epsilon$ with an optimal bridge location $1$, and let $s''$ be the solution returned by the mechanism. Note that $s''$ cannot be $s'' \geq \frac{1}{2}$, as otherwise the agent at $x''_1$ can misreport from $\frac{1}{2} + s - \epsilon$ to $0$, which will decrease this agent's cost from $\frac{1}{2} - s + \epsilon$ to $\frac{1}{2} - s - \epsilon$ (when $-s \geq \frac{1}{2} + s - \epsilon$) or $\frac{1}{2} - s - \epsilon$ (when $-s < \frac{1}{2} + s - \epsilon$). Hence, we have $s'' < \frac{1}{2}$. Also, $s''$ cannot lie in the interval $(s, -s)$, as otherwise the agent at $x_1 = 0$ in the original instance can misreport from $0$ to $\frac{1}{2} + s - \epsilon$, which will decrease the cost from $2|s|$ to $2|s''|$. Therefore, we have either $s'' \geq -s$ or $s'' \leq s$.

If $s'' \leq s$, then the social cost of the mechanism is at least $1 + \frac{1}{2} - s - \epsilon$, indicating an approximation ratio of at least $\frac{1 + \frac{1}{2} - s - \epsilon}{\frac{1}{2} - s + \epsilon} \geq \frac{\frac{3}{2} - s - \epsilon}{\frac{1}{2} - s + \epsilon} \to 2$ when $\epsilon$ approaches $0$. Therefore, we only need to consider the case $-s \leq s'' < \frac{1}{2}$.

Then we move the agent location $x''_1 = \frac{1}{2} + s - \epsilon$ to $x'''_1 = \frac{1}{2} - \delta$ for some sufficiently small positive number $\delta < \min\{\epsilon, \frac{1}{2} - s''\}$. For the instance with $x'''_1 = \frac{1}{2} - \delta$ and $y_1 = 1$, the optimal solution is $1$ and the optimal social cost is $\frac{1}{2} + \delta$. Let $s'''$ be the output of the mechanism. Note that $s'''$ cannot satisfy $s''' > \frac{1}{2}$, as otherwise this agent can misreport from $x'''_1$ to $x''_1$ and reduce the cost from $\frac{1}{2} + \delta$ to $\frac{1}{2} - \delta$. Therefore, we have $s''' \leq \frac{1}{2}$, and the social cost is at least $1 + \frac{1}{2} - \delta$, indicating an approximation ratio approaching $3$.
\end{proof}
\begin{mechanism}\label{mec3}
     For any instance $(\mathbf{x},\mathbf{y})$,  return \( s = 0 \) with probability \( \frac{n_1^2}{n_1^2 + m_2^2} \), and \( s = 1 \) with probability \( \frac{m_2^2}{n_1^2 + m_2^2} \).
\end{mechanism}
\begin{theorem}\label{thm6}
Mechanism \ref{mec3} is a randomized GSP $2$-approximation mechanism for the social cost.
\end{theorem}
\begin{proof}
We first prove that the mechanism is GSP. Suppose a coalition of agents misreports their locations to benefit all members. Note that agents in \( N_2 \) and \( M_1 \) never use the bridge and thus have no incentive to participate in misreporting. If only agents in \( N_1 \) misreport, they cannot increase the probability of \( s = 0 \); similarly, if only agents in \( M_2 \) misreport, they cannot increase the probability of \( s = 1 \). Therefore, any successful manipulation would require cooperation between agents from both \( N_1 \) and \( M_2 \). However, no misreport can simultaneously benefit both groups, as increasing the probability of \( s = 0 \) harms \( M_2 \) and increasing the probability of \( s = 1 \) harms \( N_1 \). This establishes group-strategyproofness.

Next, we analyze the approximation ratio. Recall that at least one of \( s = 0 \) or \( s = 1 \) is optimal. Since agents in \( N_2 \) and \( M_1 \) never use the bridge, the social cost depends only on \( N_1 \) and \( M_2 \). Denote 
$$A=\sum_{i \in N_1} |x_i|+\sum_{i \in M_2} |y_i|,B=\sum_{i \in N_1} |x_i - 1| + \sum_{i \in M_2} |y_i - 1|.$$
We need to show:
$$
    \frac{n_1^2}{n_1^2 + m_2^2}\cdot A + \frac{m_2^2}{n_1^2 + m_2^2}\cdot B\leq 2 \cdot \min\left(A, B\right).
$$
By symmetry, it suffices to prove one side of the inequality:
$$\frac{n_1^2}{n_1^2 + m_2^2}\cdot A + \frac{m_2^2}{n_1^2 + m_2^2}\cdot B\leq 2\left(\sum_{i \in N_1} |x_i| + \sum_{i \in M_2} |y_i|\right).
$$
Rearranging terms, we require:
\[
\frac{m_2^2}{n_1^2 + m_2^2}\cdot B \leq \frac{n_1^2 + 2m_2^2}{n_1^2 + m_2^2}\left(\sum_{i \in N_1} |x_i| + \sum_{i \in M_2} |y_i|\right).
\]
For \( i \in M_2 \), we have \( y_i > 0.5 \), so \( |y_i - 1| < |y_i| \). Thus, it is sufficient to prove:
\[
m_2^2 \sum_{i \in N_1} |x_i - 1| \leq (n_1^2 + 2m_2^2) \sum_{i \in N_1} |x_i| + (n_1^2 + m_2^2) \sum_{i \in M_2} |y_i|.
\]
Since \( y_i > 0.5 \) for \( i \in M_2 \), we have \( \sum_{i \in M_2} |y_i| \geq \frac{m_2}{2} \). Therefore, we only need:
\[
m_2^2 \sum_{i \in N_1} |x_i - 1| \leq (n_1^2 + 2m_2^2) \sum_{i \in N_1} |x_i| + \frac{n_1^2 + m_2^2}{2} m_2.
\]
For \( i \in N_1 \), we have \( |x_i - 1| \leq |x_i| + 1 \). Applying this bound:
\[
m_2^2 \sum_{i \in N_1} |x_i - 1| \leq m_2^2 \sum_{i \in N_1} |x_i| + m_2^2 n_1.
\]
Thus, it remains to show:
\[
m_2^2 \sum_{i \in N_1} |x_i| + m_2^2 n_1 \leq (n_1^2 + 2m_2^2) \sum_{i \in N_1} |x_i| + \frac{n_1^2 + m_2^2}{2} m_2.
\]
Rearranging:
\[
m_2^2 n_1 \leq (n_1^2 + m_2^2) \sum_{i \in N_1} |x_i| + \frac{n_1^2 + m_2^2}{2} m_2.
\]
Since \( |x_i| \geq 0 \) for all \( i \), it suffices to prove:
\[
m_2^2 n_1 \leq \frac{n_1^2 + m_2^2}{2} m_2.
\]
This inequality is equivalent to:
\[
2 m_2 n_1 \leq n_1^2 + m_2^2.
\]  

Example \ref{eg4} shows the upper bound to be tight.
\begin{example}\label{eg4}
     Consider the instance \((\textbf{x},\textbf{y})=(x_{1}, y_{1})\), where \(x_{1} = \frac{1}{2}-\epsilon\) and \(y_{1} = 1\) where $\epsilon\rightarrow0^+$. The optimal social cost is \( OPT= \frac{1}{2}\) and the social cost of Mechanism \ref{mec3} is \(ALG =\frac{1}{2} \times \left(\frac{3}{2}-\epsilon\right) + \frac{1}{2} \times \left(\frac{1}{2}+\epsilon\right) = 1\), so the ratio approaches 2 as $\epsilon\rightarrow0^+$.
\end{example}
This completes the proof of the approximation ratio
\end{proof}
\begin{proposition}
If an SP randomized mechanism’s output distribution depends only on \( n_1, n_2, m_1, m_2 \), then no such mechanism can achieve an approximation ratio better than 2 for the social cost.
\end{proposition}








Before proving the lower bound for any randomized strategyproof mechanism, recall that strategyproofness is equivalent to partial group-strategyproofness for facility location games where each agent has her location as a private location \cite{LW8}. In other words, for any group of agents located at the same location, no member can benefit if they misreport simultaneously.
\begin{theorem}\label{thx}
    Any randomized SP mechanism has an approximation ratio of at least $1.1$ for the social cost.
\end{theorem}
\subsection{SGSP Mechanism for Social Cost}
\begin{mechanism}\label{mec6}
For any instance $(\mathbf{x},\mathbf{y})$:
\begin{itemize}
    \item When \( n \geq m \): if \( n_1 > 0 \), return \( s = 0 \); else return \( s = 1 \).
    \item When \( m > n \): if \( m_2 > 0 \), return \( s = 1 \); else return \( s = 0 \).
\end{itemize}
\end{mechanism}
\begin{theorem}
Mechanism \ref{mec6} is a deterministic SGSP \( (1 + 2\min\{m,n\})\)-approximation mechanism for the social cost.
\end{theorem}
\begin{proof}
Firstly we prove it is SGSP. Due to symmetry, we consider the case when \( n \geq m \):

\textbf{Case 1.} If \( n_1 > 0 \) and the mechanism returns \( s = 0 \), suppose some agents misreport to change the outcome to \( s = 1 \). Then all agents in \( N_1 \) must participate in misreporting. However, for any \( i \in N_1 \), we have \( x_i < 0.5 \), so the cost under \( s = 0 \) is \( x_i \), while under \( s = 1 \) it is \( 1 - x_i > 0.5 > x_i \). Thus, all participating agents in \( N_1 \) would strictly increase their costs.

\textbf{Case 2.} If \( n_1 = 0 \) and the mechanism returns \( s = 1 \), suppose some agents misreport to change the outcome to \( s = 0 \). Then agents in \( N_2 \) must participate. For any \( i \in N_2 \), \( x_i \geq 0.5 \), so the cost under \( s = 1 \) is \( 1 - x_i \leq 0.5 \), while under \( s = 0 \) it is \( x_i \geq 0.5 \). Thus, no agent in \( N_2 \) can strictly decrease their cost.

Next we analyze the approximation ratio. Let ALG and OPT denote the social costs of the mechanism and the optimal solution, respectively. Due to symmetry, we only consider \( n \geq m \) and \( s = 0 \). If \( s = 1 \) (which means \( n_1 = 0 \)), it is optimal. Assume the optimal solution is \( s^* = 1 \) (the optimal solution is either 0 or 1). Then:

\begin{align*}
\frac{ALG}{OPT} &\leq \frac{\sum_{i \in N_1} |x_i| + \sum_{i \in N_2} |x_i - 1| + \sum_{i \in M} |y_i|}{\sum_{i \in M_1} |y_i| + \sum_{i \in M_2} |y_i - 1| + \sum_{i \in N} |x_i - 1|}\\
&\leq \frac{\sum_{i \in N_1} |x_i| + \sum_{i \in M_2} |y_i|}{\sum_{i \in M_2} |y_i - 1| + \sum_{i \in N_1} |x_i - 1|}\\
&\leq \frac{\sum_{i \in N_1} |x_i| + \sum_{i \in M_2} (|y_i - 1| + 1)}{\sum_{i \in N_1} |x_i - 1| + \sum_{i \in M_2} |y_i - 1|}\\
&\leq \frac{n_1 \cdot \frac{1}{2} + m_2 \cdot 1}{n_1 \cdot (1 - \frac{1}{2}) + m_2 \cdot 0}\\
&= 1 + 2m_2 \leq 1 + 2m.
\end{align*}
Since \( m \leq n \) in this case, the approximation ratio is bounded by \( 1 + 2m \leq 1 + 2\min(m,n) \). By symmetry, the same bound holds for the case \( m > n \).
\end{proof}
\begin{theorem}
Any deterministic SGSP mechanism has an approximation ratio of at least \(1 + \min\{m, n\}\) for the social cost.
\end{theorem}
\begin{mechanism}\label{mec7}
For any instance $(\mathbf{x},\mathbf{y})$:
\begin{itemize}
    \item If \( n_1 = 0 \), return \( s = 1 \);
    \item Else if \( m_2 = 0 \), return \( s = 0 \);
    \item Else (\( n_1, m_2 > 0 \)) return \( s = 0 \) with probability \( \frac{n}{m+n} \) and \( s = 1 \) with probability \( \frac{m}{m+n} \).
\end{itemize}
\end{mechanism}
\begin{theorem}
    Mechanism \ref{mec7} is a randomized SGSP \(\left(1 + \frac{2mn}{m+n}\right)\)-approximation mechanism for the social cost.
\end{theorem}

Based on Theorem \ref{thx}, combined with the definitions of SP and SGSP, the following result are derived.
\begin{theorem}
    Any randomized SGSP mechanism has an approximation ratio of at least 1.1 for the social cost.
\end{theorem}
\section{Conclusions and Open Problems}
In this paper, we studied the mechanism design problem for locating a bridge between two separate regions, each of which has been equipped with a facility. Each agent has her location as private information and can be served by the nearest facility. We derived upper and lower bounds on the approximation ratio for deterministic and randomized strategyproof and (strong) group-strategyproof mechanisms.

For future work, a natural direction is to narrow the gap between the upper bound and the lower bound for strategyproof mechanisms in our model. Moreover, our model primarily focuses on the impact of connectivity itself on decision-making, hence the assumption of no additional cost for crossing the bridge. Building upon this foundation, introducing bridge tolls would allow for the development of more realistic and challenging variants of the problem. Another potential extension lies in relaxing the assumption of facility homogeneity. In real-world scenarios, facilities often differ in quality, accessibility, or capacity, and agents may exhibit heterogeneous preferences. Lastly, while this work considers a simple spatial setting with two regions separated by a barrier, it would be valuable to generalize the model to more complex topologies, such as arbitrary metric spaces or network structures.

\appendix
\renewcommand{\thetheorem}{\Alph{section}.\arabic{theorem}}
\renewcommand{\theproposition}{\Alph{section}.\arabic{proposition}}
\renewcommand{\thelemma}{\Alph{section}.\arabic{lemma}}
\renewcommand{\theremark}{\Alph{section}.\arabic{remark}}
\renewcommand{\thedefinition}{\Alph{section}.\arabic{definition}}
\renewcommand{\themechanism}{\Alph{section}.\arabic{mechanism}}
\renewcommand{\theexample}{\Alph{section}.\arabic{example}}
\providecommand{\theHtheorem}{}
\providecommand{\theHproposition}{}
\providecommand{\theHlemma}{}
\providecommand{\theHremark}{}
\providecommand{\theHdefinition}{}
\providecommand{\theHmechanism}{}
\providecommand{\theHexample}{}
\renewcommand{\theHtheorem}{appendix.\Alph{section}.\arabic{theorem}}
\renewcommand{\theHproposition}{appendix.\Alph{section}.\arabic{proposition}}
\renewcommand{\theHlemma}{appendix.\Alph{section}.\arabic{lemma}}
\renewcommand{\theHremark}{appendix.\Alph{section}.\arabic{remark}}
\renewcommand{\theHdefinition}{appendix.\Alph{section}.\arabic{definition}}
\renewcommand{\theHmechanism}{appendix.\Alph{section}.\arabic{mechanism}}
\renewcommand{\theHexample}{appendix.\Alph{section}.\arabic{example}}
\setcounter{theorem}{0}
\setcounter{proposition}{0}
\setcounter{lemma}{0}
\setcounter{remark}{0}
\setcounter{definition}{0}
\setcounter{mechanism}{0}
\setcounter{example}{0}
\input{supplementary}

\bibliographystyle{named}
\bibliography{ijcai26}

\end{document}

%% file: supplementary.tex
\section{Supplementary Material: Proofs of Partial Results}\label{app:supplementary}
\section*{Maximum Cost}
\begin{proposition}
    \label{app:thm1}
For any instance $(\mathbf{x},\mathbf{y)}$, an optimal solution for minimizing the maximum cost is given by  \begin{equation*}
    s^{*}( \mathbf{x},\mathbf{y})=\begin{cases}
        0, & \mbox{if }\textbf{ }1-lt(\mathbf{x})\geq rt(\mathbf{y}), \\
        1, & \mbox{if }\textbf{ }1-lt(\mathbf{x})< rt(\mathbf{y}).
      \end{cases}
  \end{equation*}
\end{proposition}
\begin{proof}
    For any instance $(\mathbf{x},\mathbf{y})$, without loss of generality, let $1-lt(\mathbf{x})\geq rt(\mathbf{y})$. It is evident that regardless of the location of the bridge, each agent $i\in N_{2}$ will choose the facility located at $x_{F}$ for service, while each agent $j\in M_{1}$ will choose the facility located at $y_{F}$. Let
\begin{align*}
    a &= \max_{i \in N_{2}, j \in M_{1}} \left\{c_{1}(x_{i}, s), c_{2}(y_{j}, s)\right\},\\
    b &= \max\{x_{i}|i\in N_{1}\text{ }and\text{ }x_{i}\geq0\}.
\end{align*}
Furthermore, if $s^{*}(\mathbf{x},\mathbf{y})\in(-\infty,lt(\mathbf{x}))\cup(\frac{1}{2},\infty)$, it is easy to prove 
$$MC((\mathbf{x},\mathbf{y}),0)\leq MC((\mathbf{x},\mathbf{y}),s^{*}).$$
Therefore, $s^{*}(\mathbf{x},\mathbf{y})\in [lt(x),\frac{1}{2}].$ 

Let $s^{*}(\mathbf{x},\mathbf{y})=\frac{1}{2}-\epsilon$. If $\epsilon\in[0,\frac{1}{2}]$, 
\begin{align*}
    &MC((\mathbf{x}, \mathbf{y}), \frac{1}{2}-\epsilon) \\&= \max\left\{a,b,  rt(\mathbf{y}),\frac{1}{2}-\epsilon-lt(\mathbf{x})+\frac{1}{2}-\epsilon\right\} \\
    &=\max\left\{a,b, rt(\mathbf{y}),1-lt(\mathbf{x})-2\epsilon\right\}.
\end{align*}
Then, we have
$$\arg\max_{\epsilon\in[0,\frac{1}{2}]}MC((\mathbf{x}, \mathbf{y}), \frac{1}{2}-\epsilon)=\frac{1}{2}.$$
If $\epsilon\geq\frac{1}{2}$, 
\begin{align*}
    &MC((\mathbf{x}, \mathbf{y}), \frac{1}{2}-\epsilon)\\&= \max\left\{a,b-(\frac{1}{2}-\epsilon)+0-(\frac{1}{2}-\epsilon), rt(\mathbf{y}),0-lt(\mathbf{x})\right\} \\
    &= \max\left\{a,b-1+2\epsilon, rt(\mathbf{y}),-lt(\mathbf{x})\right\}.
\end{align*}
Then, we have
$$\arg\max_{\epsilon\geq\frac{1}{2}}MC((\mathbf{x}, \mathbf{y}), \frac{1}{2}-\epsilon)=\frac{1}{2}.$$
In conclusion, the optimal solution is $s^{*}(\mathbf{x},\mathbf{y})=0$. If $1-lt(\mathbf{x})< rt(\mathbf{y})$, the proof is similar.
\end{proof}
\subsection{SGSP Mechanism for Maximum Cost}
\setcounter{mechanism}{0}
\setcounter{mechanism}{1}
\setcounter{theorem}{1}
\begin{mechanism}\label{app:mec4}
For any instance $(\mathbf{x},\mathbf{y})$, if \( n_1 > 0 \), return \( s = 0 \); otherwise return \( s = 1 \).
\end{mechanism}
\begin{theorem}
Mechanism \ref{app:mec4} is a deterministic SGSP 3-approximation mechanism for the maximum cost.
\end{theorem}
\begin{proof}
We first prove the mechanism is SGSP. Consider two cases:

\textbf{Case 1.} If \( n_1 > 0 \) and the mechanism returns \( s = 0 \), suppose some agents misreport to change the outcome to \( s = 1 \). Then all agents in \( N_1 \) must participate in the misreporting. However, for any \( i \in N_1 \), we have \( x_i < 0.5 \), so the cost under \( s = 0 \) is \( x_i \), while under \( s = 1 \) it is \( 1 - x_i > 0.5 > x_i \). Thus, all agents in \( N_1 \) would strictly increase their costs, contradicting the definition of SGSP.

\textbf{Case 2.} If \( n_1 = 0 \) and the mechanism returns \( s = 1 \), suppose some agents misreport to change the outcome to \( s = 0 \). Then agents in \( N_2 \) must participate. For any \( i \in N_2 \), \( x_i \geq 0.5 \), so the cost under \( s = 1 \) is \( 1 - x_i \leq 0.5 \), while under \( s = 0 \) it is \( x_i \geq 0.5 \). Thus, no agent in \( N_2 \) can strictly decrease their cost, again contradicting SGSP.
Therefore, the mechanism is SGSP.
We now analyze the approximation ratio. If \( n_1 = 0 \), the mechanism returns \( s = 1 \), which is optimal. Consider the case where \( n_1 > 0 \) and the mechanism returns \( s = 0 \), while the optimal solution is \( s^* = 1 \) (if \( s^* = 0 \) is optimal, then the mechanism is optimal).

For any instance $(\mathbf{x},\mathbf{y})$, let ALG and OPT denote the maximum costs of the mechanism and the optimal solution,
respectively.we have:
\begin{align*}
     ALG &= \max\left(\max_{i \in N} \min(|x_i - 0|, |x_i - 1|), \max_{i \in M} |y_i - 0|\right),\\
     OPT &= \max\left(\max_{i \in M} \min(|y_i - 0|, |y_i - 1|), \max_{i \in N} |x_i - 1|\right).
\end{align*}
Expanding these expressions using the agent partitions,
\begin{align*}
    ALG &= \max\left(\max_{i \in N_1} |x_i|, \max_{i \in N_2} |x_i - 1|, \max_{i \in M} |y_i|\right),\\
    OPT &= \max\left(\max_{i \in N} |x_i - 1|, \max_{i \in M_1} |y_i|, \max_{i \in M_2} |y_i - 1|\right).
\end{align*}
We can bound the ratio as follows,
\begin{align*}
\frac{ALG}{OPT} &\leq \frac{\max\left(\max_{i \in N_1} |x_i|, \max_{i \in M_2} |y_i|\right)}{\max\left(\max_{i \in N_1} |x_i - 1|, \max_{i \in M_2} |y_i - 1|\right)} \\
&= \frac{\max\left(\max_{i \in N_1} x_i, \max_{i \in M_2} y_i\right)}{\max\left(1 - \min_{i \in N_1} x_i, \max_{i \in M_2} (1 - y_i)\right)}\\
&\leq \frac{\max\left(\frac{1}{2}, \max_{i \in M_2} y_i\right)}{\max\left(1 - \frac{1}{2}, \max_{i \in M_2} (1 - y_i)\right)} \\
&= \frac{\max_{i \in M_2} y_i}{\max\left(\frac{1}{2}, 1 - \min_{i \in M_2} y_i, \max_{i \in M_2} y_i - 1\right)}\\
&\leq \frac{\max_{i \in M_2} y_i}{\max\left(\frac{1}{2}, \max_{i \in M_2} y_i - 1\right)}.
\end{align*}
Let \( A = \max_{i \in M_2} y_i \) and note that \( A > 0.5 \). Then:
\[
\frac{ALG}{OPT} \leq \frac{A}{\max\left(\frac{1}{2}, A - 1\right)}.
\]

We consider two cases:

\textbf{Case 1.} \( A \leq \frac{3}{2} \). Then \( A - 1 \leq \frac{1}{2} \), so the denominator is \(\frac{1}{2}\). Thus, the ratio is at most \( A/(1/2) = 2A \leq 3 \).

\textbf{Case 2.} \( A > \frac{3}{2} \). Then the denominator is \( A - 1 \). The function \( f(A) = A/(A - 1) \) is decreasing for \( A > 1 \), and when \( A = 1.5 \), \( f(A) = 3 \). As \( A \) increases beyond 1.5, the ratio decreases toward 1. 

The upper bound is shown to be tight by Example \ref{app:eg1}.

\begin{example}\label{app:eg1}
    Consider the instance \((\textbf{x},\textbf{y})=(x_{1}, y_{1})\), \(x_{1} = \frac{1}{2}-\epsilon\) and \(y_{1} = \frac{3}{2}\) where $\epsilon \to 0^+$. The optimal maximum cost is $OPT = \max\left\{\min\{|y_1 - 0|, |y_1 - 1|\}, |x_1 - 1|\right\} = \max\left(\frac{1}{2}, \frac{1}{2} + \epsilon\right) = \frac{1}{2} + \epsilon,$ the maximum cost of Mechanism \ref{app:mec4} is $ALG= \max\left\{\min\{|x_1 - 0|, |x_1 - 1|\}, |y_1 - 0|\right\} = \max\left\{\frac{1}{2}, \frac{3}{2}\right\} = \frac{3}{2},
    $
so the ratio approaches 3 as \( \epsilon \to 0^+ \).
\end{example} 
Therefore, in the worst case, the approximation ratio is at most 3. 
\end{proof}
\begin{theorem}
Any deterministic SGSP mechanism has an approximation ratio of at least $2$ for the maximum cost.
\end{theorem}
\begin{proof}
Suppose there exists an SGSP mechanism \( f \) with approximation ratio better than 2. Consider an instance with one agent on \( L_1 \) at \( x_1 = 0 \) and one agent on \( L_2 \) at \( y_1 = 1 \). We analyze three cases based on the output \( s = f(x_1, y_1) \):

\textbf{Case 1.} \( s \notin \{0, 1\} \). By symmetry, assume \( s \leq \frac{1}{2} \). Now consider the scenario where agent \( y_1 \) misreports her location from 1 to 0. In this new profile, the optimal maximum cost becomes 0 (achieved by \( s^* = 0 \)). To maintain an approximation ratio better than 2, the mechanism must output \( s' = 0 \). However, this violates SGSP: after the misreporting, agent \( x_1 \) benefits (her cost decreases from \( 2s \) to 0), while agent \( y_1 \) incurs no loss (her cost remains 1 in both cases).

\textbf{Case 2.} \( s = 0 \). Now consider moving agent \( x_1 \) to \( x_1' = \frac{1}{2} - \epsilon \), where \( \epsilon \to 0^+ \). By SGSP, we cannot have \( s' > \frac{1}{2} + \epsilon \) in the new profile; otherwise, agent \( x_1' \) could benefit by misreporting back to \( x_1 = 0 \) to change the outcome to \( s = 0 \). Therefore, the maximum cost for agents \( x_1' \) and \( y_1 \) is 1, while the optimal maximum cost is \( \frac{1}{2} + \epsilon \) (achieved by \( s^* = 1 \)). This yields an approximation ratio of at least \( \frac{1}{\frac{1}{2} + \epsilon} \to 2 \) as \( \epsilon \to 0^+ \).

\textbf{Case 3.} \( s = 1 \). By symmetry, a similar argument shows the approximation ratio approaches 2.

In all cases, we reach a contradiction with the assumption that the mechanism has an approximation ratio better than 2. Therefore, no deterministic SGSP mechanism can achieve an approximation ratio better than 2 for the maximum cost.
\end{proof}
\begin{mechanism}\label{app:mec5}
For any instance $(\mathbf{x},\mathbf{y})$:
\begin{itemize}
    \item If \( n_1 = 0 \), return \( s = 1 \);
    \item Else if \( m_2 = 0 \), return \( s = 1 \);
    \item Else return \( s = 0 \) with probability \( \frac{1}{2} \) and \( s = 1 \) with probability \( \frac{1}{2} \).
\end{itemize}
\end{mechanism}
\begin{theorem}
Mechanism \ref{app:mec5} is a randomized SGSP 2-approximation mechanism for the maximum cost.
\end{theorem}
\begin{proof}
We first prove the mechanism is SGSP. For any instance $(\mathbf{x},\mathbf{y})$, if \( n_1 = 0 \) or \( m_2 = 0 \), the mechanism returns the optimal solution \( s \), and no agent can benefit from misreporting. Now consider \( n_1 > 0 \) and \( m_2 > 0 \). The mechanism randomizes between \( s = 0 \) and \( s = 1 \). To change the outcome, agents from both \( N_1 \) and \( M_2 \) must misreport. However, if the outcome changes to always \( s = 1 \), then agents in \( N_1 \) suffer increased costs; if it changes to always \( s = 0 \), then agents in \( M_2 \) suffer increased costs. Thus, no group can misreport to benefit some members without harming others, satisfying SGSP.

For the approximation ratio, if \( n_1 = 0 \) or \( m_2 = 0 \), the mechanism is optimal. Now consider \( n_1 > 0 \) and \( m_2 > 0 \). Let ALG and OPT denote the maximum costs of the mechanism and the optimal solution, respectively. Without loss of generality, assume the optimal solution is \( s^* = 1 \). From Mechanism \ref{app:mec4} we know that when \( s = 0 \), the cost is at most \( 3 \cdot OPT \), and when \( s = 1 \), the cost is at most OPT. Thus:
$$    ALG \leq \frac{1}{2} \cdot 3OPT + \frac{1}{2} \cdot OPT = 2 \cdot OPT.
$$
Example \ref{app:eg2} shows the upper bound to be tight.
\begin{example}\label{app:eg2}
    Consider the instance \((\textbf{x},\textbf{y})=(x_{1}, y_{1})\), \(x_{1} = \frac{1}{2}-\epsilon\) and \(y_{1} = \frac{3}{2}\) where $\epsilon \rightarrow 0^+$. The optimal maximum cost is $
   OPT  = \frac{1}{2},$ the maximum cost of Mechanism \ref{app:mec5} is $ ALG=\frac{1}{2}\cdot\left(\frac{1}{2}+\epsilon\right)+\frac{1}{2}\cdot\frac{3}{2}= 1+\frac{\epsilon}{2},$
so the ratio approaches 2 as $\epsilon\rightarrow 0^+$.  
\end{example} 
Therefore, the mechanism achieves a 2-approximation. 
\end{proof}
\section*{Social Cost}
\begin{proposition}
\label{app:t1}
For the social cost, at least one of the endpoints $s\in \{0,1\}$ is an optimal bridge location. 
\end{proposition}
\begin{proof}
    For any instance \((\mathbf{x}, \mathbf{y})\), it is evident that regardless of the bridge location, each agent $i\in N_{2}$ will choose facility $x_{F}$, while each agent $j\in M_{1}$ will select facility $y_{F}$. For any solution $s \in R$, we will analyze the social cost in two cases: $s \leq 0.5$ and $s > 0.5$.

\textbf{Case 1.} $s \leq 0.5$.
For each agent $ i \in N_1$,
\begin{equation*}
c_1(x_i,0) = |x_i - 0| \leq |x_i - s| + |s| = c_1(x_i,s)
\end{equation*}
Similarly, for each agent $ j \in M_2$, $c_2(y_j,0) \leq c_2(y_j,s)$. Therefore,
\begin{align*}
SC((\mathbf{x}, \mathbf{y}),s) =& \sum_{i \in N_1} c_1(x_i,s) + \sum_{i \in N_2} c_1(x_i,s) \\&+ \sum_{j \in M_1} c_2(y_j,s) + \sum_{j \in M_2} c_2(y_j,s) \\
\geq& \sum_{i \in N_1} c_1(x_i,0) + \sum_{i \in N_2} c_1(x_i,0) \\&+ \sum_{j \in M_1} c_2(y_j,0) + \sum_{j \in M_2} c_2(y_j,0) \\
=& SC((\mathbf{x}, \mathbf{y}),0).
\end{align*}
Thus, selecting $s=0$ achieves the minimum social cost when $s \leq 0.5$.

\textbf{Case 2.} $s > 0.5$.
For each agent $i \in N_1$,
\begin{align*}
c_1(x_i,s) &= \min\left\{|1-x_i|,|x_i - s| + |s-1|\right\}\\&\geq |1-x_i|=c_1(x_i,s).
\end{align*}
Similarly, for each agent $j \in M_2$, $c_2(y_j,s) \geq c_2(y_j,1)$. Therefore,
\begin{align*}
SC((\mathbf{x}, \mathbf{y}),s) =& \sum_{i \in N_1} c_1(x_i,s) + \sum_{i \in N_2} c_1(x_i,s) \\&+\sum_{j \in M_1} c_2(y_j,s) + \sum_{j \in M_2} c_2(y_j,s) \\
\geq& \sum_{i \in N_1} c_1(x_i,1) + \sum_{i \in N_2} c_1(x_i,1) \\&+ \sum_{j \in M_1} c_2(y_j,1) + \sum_{j \in M_2} c_2(y_j,1) \\
=& SC((\mathbf{x}, \mathbf{y}),1).
\end{align*}
Hence,  selecting $s=1$ achieves minimum the social cost when $s> 0.5$.
In conclusion, the theorem guarantees the existence of an optimal solution.
\end{proof}
\subsection{GSP Mechanism for Social Cost}
\setcounter{mechanism}{3}
\setcounter{theorem}{5}
\begin{proposition}
If a deterministic SP mechanism only considers \( n_1, n_2, m_1, m_2 \) instead of the complete agents' profile, then no such mechanism is better than 3-approximation.
\end{proposition}
\begin{proof}
Let \( n_1 = m_2 = 1 \), and \( n_2 = m_1 = 0 \). Such a mechanism must output the same \( s \) for the following two profiles: \( x_1 = 0, y_1 = 0.5 + \epsilon \) and \( x_1 = 0.5 - \epsilon, y_1 = 1 \), where \( \epsilon \to 0^+ \) is a very tiny constant. Omitting the value of \( \epsilon \), we know the optimal social cost for both profiles is the same, i.e., 0.5. Due to symmetry, we assume \( s \leq \frac{1}{2} \). However, in this case, the social cost of the second profile will be \( 0.5 + 1 = 1.5 \), indicating that the lower bound is at least a 3-approximation.
\end{proof}

\begin{proposition}
If an SP randomized mechanism’s output distribution depends only on \( n_1, n_2, m_1, m_2 \), then no such mechanism can achieve an approximation ratio better than 2 for the social cost.
\end{proposition}
\begin{proof}
Consider the case where \( n_1 = m_2 = 1 \) and \( n_2 = m_1 = 0 \). We analyze two instances:  
\( x_1 = \frac{1}{2} - \epsilon, y_1 = 1; \) and \( x_1 = 0, y_1 = \frac{1}{2} + \epsilon \).

For both instances, the optimal social cost is \(\frac{1}{2}\). Since the mechanism depends only on the counts \( n_1, m_2 \), it must assign identical probabilities to \( s = 0 \) and \( s = 1 \) in both cases. Let \( p \) denote the probability of selecting \( s = 0 \).

In Case 1, the expected social cost is:

\[p \cdot \frac{3}{2} + (1 - p) \cdot \frac{1}{2} = \frac{1}{2} + p.\]

To achieve an approximation ratio better than 2, we require \(\frac{1}{2} + p < 1\), which implies \( p < \frac{1}{2} \).

In Case 2, the expected social cost is:

\[p \cdot \frac{1}{2} + (1 - p) \cdot \frac{3}{2} = \frac{3}{2} - p.\]

For an approximation ratio better than 2, we require \(\frac{3}{2} - p < 1\), which implies \( p > \frac{1}{2} \).

These two requirements are contradictory. Therefore, no such mechanism can achieve an approximation ratio better than 2.
\end{proof}
\setcounter{theorem}{7}
\begin{theorem}\label{app:thx}
    Any randomized SP mechanism has an approximation ratio of at least $1.1$ for the social cost.
\end{theorem}
\begin{proof}
    Assume that there exists a randomized strategyproof mechanism \( f \) with an approximation ratio less than \( 1.1 \).  Consider the instance \((\mathbf{x}, \mathbf{y})\) with $x_i=0.24$ for all agents $i\in N$, where $|N|=4$, and $y_j=1$ for all agent $j\in M$, where $|M|=3.$ It is clear that \(s^{*}(\mathbf{x}, \mathbf{y}) = 1\) and \(SC((\mathbf{x}, \mathbf{y}), s^{*}) = 4 \times 0.76 = 3.04\). Suppose \(f(\mathbf{x},\mathbf{y})\) outputs \(Y \geq 0.5\) with probability \(p\). Then by the proof of Theorem \ref{app:t1},
\begin{align*}
    &SC((\mathbf{x}, \mathbf{y}), f)\\ &\geq p \cdot SC((\mathbf{x}, \mathbf{y}), 1) + (1 - p) \cdot SC((\mathbf{x}, \mathbf{y}), 0) \\
    &= 3.04p + 3.96 \cdot (1 - p) = 3.96 - 0.92p. \\
    &\frac{SC((\mathbf{x}, \mathbf{y}), f)}{SC((\mathbf{x}, \mathbf{y}), s^{*})} < 1.1 \\&\Rightarrow 3.96 - 0.92p < 1.1 \times 3.04 \\&\Rightarrow p > 0.66.
\end{align*}
Therefore,
\[
c_{1}(x_{1}, f(\mathbf{x}, \mathbf{y})) \geq 0.76 p + 0.24 \cdot (1 - p) > 0.58.
\]

Now Consider the instance \((\mathbf{x}', \mathbf{y})\) with $x_i'=0$ for all agents $i\in N$, where $|N|=4$. It is clear that \(s^{*}(\mathbf{x'}, \mathbf{y}) = 0\) and \(SC((\mathbf{x'}, \mathbf{y}), s^{*}) = 3\). Suppose \(f(\mathbf{x'}, \mathbf{y})\) outputs \(Y \in [-0.12, 0.12]\) with probability \(q\). Then
\begin{align*}
    &SC((\mathbf{x'}, \mathbf{y}), f) \\\geq& (1 - q) \cdot \min\{SC((\mathbf{x'}, \mathbf{y}), -0.12), SC((\mathbf{x'}, \mathbf{y}), 1)\} \\&+q \cdot SC((\mathbf{x'}, \mathbf{y}), 0)\\
    =& 3 q + 3.96 \cdot (1 - q) = 3.96 - 0.96 q.\\
 &\frac{SC((\mathbf{x'}, \mathbf{y}), f)}{SC((\mathbf{x'}, \mathbf{y}), s^{*})} < 1.1 \\\Rightarrow &3.96 - 0.96 q < 1.1 \times 3\\ \Rightarrow &q > 0.68.
 \end{align*}
Therefore,
\[
c_{1}(x_{1}, f(\mathbf{x'}, \mathbf{y})) \leq 0.48 q + 0.76 \cdot (1 - q) < 0.57.
\]
Since \(
c_{1}(x_{1}, f(\mathbf{x'}, \mathbf{y})) < c_{1}(x_{1}, f(\mathbf{x}, \mathbf{y})),
\) 
agent $i$ with $x_i=0.24$ can benefit by misreporting their location as $0$. This contradicts the strategyproofness of the mechanism, proving that the assumption does not hold.
\end{proof}
\subsection{SGSP Mechanism for Social Cost}
\setcounter{theorem}{9}
\setcounter{mechanism}{6}
\begin{theorem}
Any deterministic SP mechanism has an approximation ratio of at least \(1 + \min\{m, n\}\) for the social cost.
\end{theorem}
\begin{proof}
Consider an instance with \(n\) agents on \(l_1\) at 0 and \(m\) agents on \(l_2\) at 1. Suppose there exists an SGSP mechanism \(f\) with approximation ratio better than \(1 + \min\{m, n\}\). We analyze three cases based on the output \(s\):

\textbf{Case 1.} \(s \notin \{0, 1\}\). By symmetry, assume \(s \leq \frac{1}{2}\). Now consider the scenario where agents in \(M\) misreport their locations from 1 to 0. In this new profile, the optimal social cost becomes 0 (achieved by \(s^* = 0\)). To maintain an approximation ratio better than \(1 + \min\{m, n\}\), the mechanism must output \(s' = 0\). However, this violates SGSP: after the misreporting, agents in \(N\) benefit (their cost decreases from \(2s\) to 0), while agents in \(M\) incur no loss (their cost remains 1 in both cases).

\textbf{Case 2.} \(s = 0\). Now consider moving all agents in \(N\) except \(x_1\) to \(x_i' = 1\). By SGSP, we cannot have \(s' > \frac{1}{2} + \epsilon\) in the new profile; otherwise, these agents in \(N\) except \(x_1\) could benefit by misreporting back to 0 to change the outcome to \(s = 0\). Therefore, the social cost for agents is at least \(m + 1\), while the optimal social cost is 1 (achieved by \(s^* = 1\)). This yields an approximation ratio of at least \(1 + m \geq 1 + \min\{m, n\}\).

\textbf{Case 3.} \(s = 1\). By symmetry, a similar argument shows the approximation ratio approaches \(1 + n \geq 1 + \min\{m, n\}\).

In all cases, we reach a contradiction with the assumption that the mechanism has an approximation ratio better than \(1 + \min\{m, n\}\). Therefore, no deterministic SGSP mechanism can achieve an approximation ratio better than \(1 + \min\{m, n\}\) for the social cost.
\end{proof}
\begin{mechanism}\label{app:mec7}
For any instance $(\mathbf{x},\mathbf{y})$:
\begin{itemize}
    \item If \( n_1 = 0 \), return \( s = 1 \);
    \item Else if \( m_2 = 0 \), return \( s = 0 \);
    \item Else (\( n_1, m_2 > 0 \)) return \( s = 0 \) with probability \( \frac{n}{m+n} \) and \( s = 1 \) with probability \( \frac{m}{m+n} \).
\end{itemize}
\end{mechanism}
\begin{theorem}
    Mechanism \ref{app:mec7} is a randomized SGSP \(\left(1 + \frac{2mn}{m+n}\right)\)-approximation mechanism for the social cost.
\end{theorem}
\begin{proof}
 We first prove it is SGSP. If \( n_1 = 0 \) or \( m_2 = 0 \), the mechanism returns the optimal solution \( s \), and no agent can benefit from misreporting. Now consider \( n_1 > 0 \) and \( m_2 > 0 \). The mechanism randomizes between \( s = 0 \) and \( s = 1 \). To change the outcome, agents from both \( N_1 \) and \( M_2 \) must misreport. However, if the outcome changes to increase the probability of \( s = 1 \), then agents in \( N_1 \) suffer increased costs; if it changes to increase the probability of \( s = 0 \), then agents in \( M_2 \) suffer increased costs. Thus, no group can misreport to benefit some members without harming others.

For the approximation ratio, if \( n_1 = 0 \) or \( m_2 = 0 \), the mechanism is optimal. Now consider \( n_1 > 0 \) and \( m_2 > 0 \). let ALG and OPT denote the social costs of the mechanism and the optimal solution, respectively. Without loss of generality, assume \( n \geq m \). If the optimal solution is \( s^* = 1 \), from Mechanism \ref{mec6} we know that when \( s = 0 \), the cost is at most \( (1 + 2m) \cdot OPT\), and when \( s = 1 \), the cost is at most OPT. Thus:
\[
\frac{ALG}{OPT} \leq \frac{n}{m+n} \cdot (1 + 2m) + \frac{m}{m+n} = 1 + 2\frac{mn}{m+n}.
\]
Similarly, if the optimal solution is \( s^* = 0 \), the same bound holds. Therefore, the mechanism achieves a \(\left(1 + 2\frac{mn}{m+n}\right)\)-approximation.
\end{proof}